\documentclass[11pt]{article}

\usepackage[english]{babel}
\usepackage{graphicx}
\usepackage{enumerate}
\usepackage{amsmath}
\usepackage{color, colortbl}

\newcommand{\vsu}{\vspace{+.1cm} }
\newcommand{\vsd}{\vspace{+.2cm} }
\newcommand{\vst}{\vspace{+.3cm} }

\newcommand{\neto}{\noindent $\bullet$ \hspace{+.17cm}}

\newtheorem{theorem}{Théorème}[section]
\newtheorem{defin}[theorem]{Définition}
\newtheorem{lemma}[theorem]{Lemme}

\newtheorem{cor}[theorem]{Corollaire}

\newtheorem{remark}[theorem]{Remarque}
\newtheorem{prop}[theorem]{Proposition}
\newtheorem{proposition}[theorem]{Proposition}

\newenvironment{proof}{\par\noindent {\bf Preuve.} \rm}{\ ~~~$\fbox{}$}

\newcommand{\Z}{\mbox{\rm \lower0.3pt\hbox{$\angle\!\!\!$}Z}}

\newcommand{\sub}[2]{#1_{[#2]}}

\newcommand{\cad}{c'est-\`a-dire}
\newcommand{\Cad}{C'est-\`a-dire}

\newcommand{\vers}{\rightarrow}

\begin{document}

 
 

\newcounter{comptnivun}
\setcounter{comptnivun}{1}
\newcounter{comptnivdeux}
\setcounter{comptnivdeux}{1}
\newcounter{comptnivtrois}
\setcounter{comptnivtrois}{1}
\newcounter{comptnivquatre}
\setcounter{comptnivquatre}{1}

\bibliographystyle{plain}

\title{Etude des morphismes préservant les mots primitifs}

\author{Francis Wlazinski}
\date{\today}

\maketitle

\abstract{Vous trouverez dans cet article un rappel de quelques propriétés sur les mots primitifs et 
les morphismes qui les préservent que l'on peut lire dans \cite{Hsi2003,Mit1996,Mit1997}. 
Leurs démonstrations que j'ai plus ou moins remaniées sont fournies. 
Ce qui fait que cet article est presque ''self contain'' comme
disent les anglais.
J'apporte aussi ma pierre à l'édifice en donnant quelques propriétés sur les mots
primitifs mais surtout en montrant qu'un morphisme sans puissance $k (\geq 5)$
est primitif et qu'un morphisme uniforme sans puissance $k (\geq 2)$ est primitif.}

\parindent=0cm
\parskip=0.15cm

\section{\label{Preliminaires} 
         Préliminaires}

Dans la suite, $A$ et $B$ sont des alphabets \cad\ des ensembles finis non vides de symboles.

Un \textit{mot} est un élément de $A^*$ le mono\"ide libre engendré par $A$ dont l'élément neutre est le
\textit{mot vide} noté $\varepsilon$ et dont la loi de composition ".", usuellement non notée, est simplement 
la juxtaposition des symboles.
On note $A^+$ l'ensemble des mots non vides \cad\ $A^+=A^* \setminus\{\varepsilon\}$.
Dans la suite, nous ne préciserons pas toujours l'alphabet utilisé car ce sera souvent peu pertinent. 

Etant donné un mot non-vide $u = a_1\ldots a_n$ avec $a_i \in A$, la \textit{longueur}
de $u$, notée $|u|$,
est le nombre entier $n$. La longueur du mot vide est $| \varepsilon | = 0$.
L'image \textit{miroir} de $u$, notée $\tilde{u}$, est le mot
$a_{n} \ldots a_{1}$. Dans le cas particulier du mot vide,
on a $\tilde{\varepsilon}=\varepsilon$.

Si on peut écrire un mot $w$ sous la forme $pw's$ alors on dit que $p, w'$ et $s$ sont des \textit{facteurs}
de $w$, que $p$ est un \textit{préfixe} de $w$ et que $s$ est un \textit{suffixe} de $w$.
En outre, si l'un des deux mots $p$ ou $s$ est non vide, alors $w'$ est dit \textit{propre}.
Si les deux mots $p$ et $s$ sont non vides, alors $w'$ est dit \textit{interne}.
Si $\varepsilon \neq p=s \neq w$, on dit que $p$ est un \textit{bord} de $w$.
De façon plus générale, un bord d'un mot $w$ est un facteur non vide de $w$ qui est à la fois préfixe propre
et suffixe propre de $w$.

Soit $w$ un mot non vide et soient $i, j$ deux entiers tels que 
$0 \leq i-1 \leq j \leq |w|$.
On note $\sub{w}{i..j}$ le facteur de $w$ tel que $|\sub{w}{i..j}|=j-i+1$ 
et $w = p \sub{w}{i..j} s$ pour deux mots $s$ et $p$ qui vérifie $|p| = i-1$.
Remarquons que, quand $j = i - 1$, nous avons $\sub{w}{i..j} = \varepsilon$.
Lorsque $i=j$, nous notons $\sub{w}{i}$ le facteur $\sub{w}{i..i}$, \cad\ la 
$i$-{ème} lettre de $w$.
En particulier, $\sub{w}{1}$ et $\sub{w}{|w|}$ sont respectivement la première et la dernière lettre de $w$.
De façon générale, $\sub{w}{i..j}$ est le facteur de $w$ qui commence à la $i$-ème lettre de $w$ et qui finit à la
$j$-ème.

Un \textit{mot infini} (ou \textit{$\omega$-mot}) sur $A$ est une suite infinie d'éléments de $A$.
L'ensemble des mots infinis sur $A$ est noté $A^{\omega}$.

Deux mots $u$ et $v$ de $A^*$ sont dit \textit{conjugués} si on peut obtenir l'un à partir de l'autre par une permutation circulaire
\cad\  s'il existe des mots $x$ et $y$ de $A^*$ tels que $u=xy$ et $v=yx$.
On dit que $v$ est un conjugué \textit{propre} de $u$ si $x$ et $y$ sont non vides.
La relation de conjugaison est trivialement une relation d'équivalence.

Dans cet article, nous ne considérons que des puissances entières : 
Les puissances d'un mot $u$ sont définies, par récurrence, par $u^0 = \varepsilon$, and 
$u^n = u u^{n-1}$ pour tout entier $n \geq 1$. 
Pour tout entier $k \geq 2$, le cas $\varepsilon^k$ est de peu d'intérêt.
Une \textit{puissance $k$} est donc un mot $u^k$ où $u\neq \varepsilon$.
On dit qu'un mot contient une puissance $k$ si l'un de ses facteurs est sous la forme $u^k$ (avec $u\neq \varepsilon$).
Un mot est dit \textit{sans puissance $k$} si la plus grande puissance qu'il contient est strictement inférieure à $k$.

Un mot est dit \textit{primitif} ou \textit{apériodique} s'il est non vide et s'il n'est pas la puissance d'un autre
mot, i.e., $w$ est primitif si l'égalité $w=u^k$, pour un entier $k$ non nul, implique $k=1$ 
(évidemment, et/ou  $w=u$).

Pour un mot $w (\neq \varepsilon)$, le plus petit (unique) mot $t$ tel que $w=t^n$ pour un entier $n \geq 1$
est appel\'e la \textit{racine primitive} de $w$. Elle est not\'ee $\rho(w)$.



\vsu

Dans la suite de cet article, on utilisera maintes fois des propriétés élémentaires de la 
combinatoire de mots. La première étant l'incontournable théorème de Fine and Wilf~:

\begin{proposition}[Fine \& Wilf]{\rm \cite{Lot1983,Lot2002}}
\label{Fine}
Si une puissance d'un mot non vide $u$ et une puissance d'un mot non vide $v$
ont un préfixe commun de longueur sup\'erieure ou \'egale \`a $|u|+|v|-pgcd(|u|,|v|)$
alors $u$ et $v$ sont les puissances d'un m\^eme mot primitif.

Autrement dit, $u$ et $v$ ont même racine primitive.
De plus, si $|u|>|v|>0$ alors $u$ n'est pas primitif.

En outre, la borne  $|u|+|v|-pgcd(|u|,|v|)$ est optimale.

\end{proposition}

\begin{remark}

Puisque $pgcd(|u|,|v|) \leq \min(|u|,|v|)$, on a en particulier que
$|u|+|v|-pgcd(|u|,|v|)\geq \max(|u|,|v|)$.

\end{remark}

En utilisant l'image miroir, on obtient directement~:

\begin{cor}\label{CorFine1}
Si une puissance d'un mot non vide $u$ et une puissance d'un mot non vide $v$
ont un suffixe commun de longueur sup\'erieure ou \'egale \`a $|u|+|v|-pgcd(|u|,|v|)$
alors $u$ et $v$ sont les puissances d'un m\^eme mot primitif.

De plus, la borne  $|u|+|v|-pgcd(|u|,|v|)$ est optimale.

\end{cor}

\begin{cor}\label{CorFine2}

Pour tous les entiers non nuls $n$ et $m$, si $u^n=v^m$ alors les mots $u$ et $v$ sont les puissances d'un même mot.

\end{cor}

\begin{remark}

Cela implique que, si $w=u^n$ avec $n \geq 1$, alors $u$ est une puissance
de $\rho(w)$ la racine primitive de $w$.

\end{remark}

\begin{cor}{\rm \cite{Ker1986}}
\label{Kera} 
Soient $x$ et $y$ deux mots.
Si une puissance de $x$ et une puissance de $y$ ont un facteur commun de longueur supérieure ou égale
à $|x|+|y|-gcd(|x|,|y|)$ alors il existe deux mots $t_1$ et $t_2$ tels que 
$x$ soit une puissance de $t_1t_2$ et $y$  soit une puissance de  $t_2t_1$ 
avec $t_1t_2$ et $t_2t_1$ des mots primitifs.
De plus, si $|x|>|y|$ alors $x$ n'est pas primitif.

\end{cor}

Les solutions des équations élémentaires sur les mots rencontrées fréquemment sont données par la 
proposition suivante~:

\begin{proposition}{\rm \cite{Lot1983}}
\label{Lothaire} 

Soient $u,v,w$ trois mots sur un alphabet $A$.
\begin{enumerate}
\topsep0cm
\itemsep0cm
\item \label{Lotcase1}
Si $vu=uw$ et $v \neq \varepsilon$, 
alors il existe deux mots $r$ et 
$s$ sur $A$, et un entier $n$ tels que $u=r(sr)^{n}$, $v=rs$
et $w=sr$.

\item \label{Lotcase2}
Si $vu=uv$, on dit que les mots $u$ et $v$ \textit{commutent}, et alors il existe un mot (primitif) 
$w$ sur $A$, et deux entiers
$n$ et $p$ tels que $u=w^{n}$ et $v=w^{p}$.

\item \label{Lotcase3}
Si $uvw=wvu$ et $(u,v) \neq (\varepsilon,\varepsilon)$
alors il existe deux mots $t_{1}$,  $t_{2}$ sur $A$
et trois nombres entiers $n$,  $p$,  $q$ tels que $u=(t_{1}t_{2}) ^ {n}t_{1}$,
$v=(t_{2}t_{1})^{p}t_{2}$ et $w=(t_{1}t_{2})^{q}t_{1}$.

\end{enumerate}
\end{proposition}

\begin{remark}\label{RemLot}
Si deux mots $x$ et $y$ non vides commutent alors $xy$ n'est pas primitif.
\end{remark}

\begin{cor}\label{NbreConjPrim}

Un mot primitif $w$ poss\`edent $|w|$ conjugu\'es diff\'erents.

\end{cor}

\begin{remark}\label{rem112}

Autrement dit, un mot est primitif si ses conjugués propres lui sont tous différents.

\end{remark}

\begin{cor}\label{Corprim112}

Tous les conjugués (propres) d'un mot primitif sont primitifs.

\end{cor}

\begin{remark}

Un mot non primitif contient un bord.

\end{remark}

\begin{proposition}

Un mot est primitif si et seulement si l'un de ses conjugués est un mot sans bord. 

\end{proposition}

\begin{proof}

Si $w$ n'est pas primitif, i.e., $w=t^n$ avec $t \neq \varepsilon$ et  $n\geq 2$, alors $t$ est un bord de $w$.

Si $w$ est primitif, soit $x$ le plus petit des conjugués de $w$ dans un ordre lexicographique fixé noté $\leq$. 
Si $x$ possède un bord alors on peut écrire $x = u v u$ pour deux mots $u \neq \varepsilon$ et $v$.
Le mot $z=uuv$ est aussi un conjugué de $w$.
Si $x = z$ alors $v u = u v$ et $x$ ne serait pas primitif~:~contraire aux hypothèses.
Donc $x < z$ et $v u < u v$. 
Mais alors $v u u < u v u=x$ avec $v u u$ un conjugué de $w$~:~c'est contraire à l'hypothèse du choix de $x$. 
\end{proof}

\begin{proposition}\label{propconjracprim}

Deux mots conjugués ont des racines primitives conjuguées.

\end{proposition}

\begin{proof}

Soient $w=rs \neq \varepsilon$ et $\underline{w}=sr$ l'un de ses conjugués.
Soit $t=\rho(w)$ la racine primitive de $w$ et soit $n$ l'entier tel que $w=t^n$.
Il existe deux mots $t_1$ et $t_2$ et un entier $\alpha$ tels que
$t=t_1t_2$, $r=t^{\alpha}t_1$ et $s=t_2t^{n-\alpha-1}$.

Autrement dit, $\underline{w}=(t_2t_1)^n$. D'après le corollaire~\ref{Corprim112}, 
le mot $t_2t_1$ est primitif. On a donc que $t_2t_1$, conjugué de $t_1t_2$, est la racine primitive
de $\underline{w}$.
\end{proof}

%

\begin{lemma}{\rm \cite{Ker1986,Lec1985}}
\label{factint}

Si un mot non vide $v$ est un \textit{facteur interne} de $vv$,
\cad\ s'il existe deux mots non vides $x$ et $y$ tels que $vv=xvy$,
alors il existe un mot non vide $t$ et deux entiers $i,j \geq 1$
tels que $x=t^i$, $y=t^j$, et $v=t^{i+j}$.
\end{lemma}

\begin{remark}
\label{rem110}

Autrement dit, si un mot non vide $v$ est un facteur interne de $vv$ alors $v$ n'est pas primitif.

\end{remark}

Un langage $X$ est un sous-ensemble de $A^*$.

Un langage $X$ est un \textit{code} si tout élément de $X^+$ admet une décomposition unique sur $X$.
Cela signifie que $X^*$ est un sous-mono\"ide libre de $A^*$ dont $X$ en est une base.

Un code $X \subset A^*$ est dit \textit{uniforme} si tous les mots non vides de $X$ ont la même longueur.

Un code est dit \textit{comma-free} s'il est uniforme et si aucun mot du code n'est facteur interne 
de la composition de deux mots du code.

Un code $X \subset A^*$ est dit \textit{pur} si la racine primitive de tout mot de $X^*$ appartient à $X^*$.

\vsu

Soient $A$ et $B$ deux alphabets.
Un \textit{morphisme} $f$ de $A^*$ vers $B^*$ est une application
de $A^*$ vers $B^*$ telle que $f(uv) = f(u)f(v)$ pour tous les mots $u, v$ de $A^*$.
Si l'alphabet $B$ n'a pas d'importance, on dira que $f$ est défini sur $A^*$. Notons qu'un morphisme
sur $A^*$ est entièrement déterminé par les images des lettres de $A$.
Le morphisme $f$ est \textit{uniforme} s'il existe un entier $L$ tel que $|f(a)|=L$ pour tout $a \in A$.

L'application $\epsilon : A^* \vers B^*; u \mapsto \varepsilon$ est un morphisme (uniforme de longueur $0$) 
qui n'a que peu d'intérêt pour nous. 
A partir de maintenant, nous supposerons donc que $f\neq \epsilon$ pour tout morphisme $f$ qui sera considéré.

\begin{remark}

Un morphisme $h$ de $A^*$ vers $B^*$ est injectif si et seulement si $h(A)$ est un code.

\end{remark}

Soit $k \geq 2$ un entier. 

Pour tout entier $n \geq 1$, un morphisme est dit  \textit{$n$-sans puissance $k$} ou 
 \textit{sans puissance $k$ jusqu'à $n$} si les images de tous les mots
sans puissance $k$ de $A^+$ de longueurs inférieures ou égales à $n$ par ce morphisme sont aussi sans puissance $k$.

Un morphisme $f$ est \textit{sans puissance $k$} s'il est $n$-sans puissance $k$ pour tout entier $n \geq 1$.
Autrement dit, un morphisme $f$ est \textit{sans puissance $k$} si
l'image de tout mot sans puissance $k$ est aussi sans puissance $k$.
On dit aussi que le morphisme $f$ préserve l'absence de puissance $k$.
Lorsque $k=2$ (resp. $k=3$), on parlera de morphismes sans carré (resp. sans cube).

De même, pour tout entier $n \geq 1$, un morphisme est dit \textit{$n$-primitif} 
ou \textit{primitif jusqu'à $n$} 
si les images de tous les mots primitifs 
de $A^+$ de longueurs inférieures 
ou égales à $n$ par ce morphisme sont aussi primitifs 

Un morphisme est dit \textit{primitif} 
s'il est $n$-primitif 
pour tout entier $n \geq 1$, i.e., 
un morphisme $f$ sur $A^*$ est primitif 
si l'image de tout mot primitif 
est elle aussi primitif 

Un morphisme $f$ sur $A^*$ est \textit{préfixe} (resp. \textit{suffixe})
si, pour toutes les lettres $a$ et $b$ différentes dans $A$, 
le mot $f(a)$ n'est pas un préfixe (resp. pas un suffixe) de $f(b)$.
Un morphisme $f$ sur $A^*$ est non-effaçant si $f(a)\neq \varepsilon$ pour toute lettre $a$ de $A$.
Un morphisme est \textit{bifixe} s'il est préfixe et suffixe.

Notons, qu'un morphisme préfixe (resp. suffixe) est injectif et non-effaçant.

Un morphisme $f$ de $A^*$ vers $B^*$ est un \textit{ps-morphisme} 
si les égalités
$f(a) = ps$, $f(b) = ps'$, et $f(c) = p's$ 
avec $a,b,c \in A$ (éventuellement $c = b$) et $p$, $s$, $p'$, et
$s' \in B^*$ implique $b = a$ ou $c = a$.
Notons qu'un ps-morphisme est bifixe.

Etant donné un morphisme $f$ sur $A$, le morphisme $\tilde{f}$
\textit{miroir}\index{morphisme!miroir}
de $f$ est défini par $\tilde{f}(a)=\widetilde{f(a)}$
pour toutes les lettres $a$ de $A$.
En particulier, on a $\tilde{f}(w)=\widetilde{f(\tilde{w})}$
pour tous les mots $w$ sur $A$.

Par exemple, si l'on considère le morphisme $f$ défini par
$f :\{a,b\}^* \rightarrow \{x,y,z\}^*$; 
$a \mapsto x y$ et $b \mapsto y z x$
Le morphisme miroir de $f$ est
donc défini par $\tilde{f}(a)= y x$ et $\tilde{f}(b)= x z y$. On
a, par exemple, $f(abb)=x y \, y z x \, y z x$. On peut vérifier que
$\tilde{f}(bba)= \tilde{f}(b) \tilde{f}(b)\tilde{f}(a) = x z y \, x z y \, y x$
$= \widetilde{f(abb)}=\widetilde{f(\widetilde{bba})}$.

Notons qu'un mot $w$ est primitif si et seulement si $\tilde{w}$ est primitif.
Comme conséquence directe, un morphisme $f$ est primitif si et seulement si
$\tilde{f}$ est primitif.

On étend naturellement la notion de morphisme aux mots infinis.

\vsu

Rappelons enfin deux résultats sur les morphismes sans puissance $k$.

\begin{proposition} \cite{Wla2015a,Wla2017} \label{puissancekunif}

Un morphisme uniforme sans puissance $k$ est sans puissance $k+1$ pour tout entier $k \geq 3$.

\end{proposition}

\begin{proposition}\cite{Wla23:hal}\label{puissancek}

Un morphisme sans puissance $k$ est sans puissance $k+1$ pour tout entier $k \geq 5$.

\end{proposition}

\section{Quelques équations}

\begin{lemma}\label{LemEqua}  {\rm{(Lemme 9 dans~\cite{Mit1997})}}

Soient $y$, $y'$, $z$ et $z'$ quatre mots (les hypothèses $|y|=|y'|$ et $|z|=|z'|$ ne sont que des conséquences des égalités).

Si $zy=y'z$ et $yz=z'y'$ alors $y=y'$.

\end{lemma}

\begin{proof}

Si $|z|=|z'|=0$, l'égalité est triviale. On suppose donc les mots $z$ et $z'$ non vides.

Si $|y| \leq |z|$ alors $y$ et $y'$ sont deux suffixes de $z$ et l'égalité est évidente.

Si $|y|>|z|$ alors il existe deux mots $y_1$ et $y_1'$ tels que 
$y'=y_1z=zy_1'$ et $y=y_1'z=z'y_1$ avec $|y_1|=|y_1'| < |y|$. 

On peut appliquer le même raisonnement avec les mots $y_1$ et $y_1'$ à la place de $y$ et $y'$.
De la même façon, on obtient soit $y_1'=y_1$ soit deux nouveaux mots $y_2$ et $y_2'$.
On continue ainsi de suite. On obtient deux suites strictement décroissantes (en longueur) de mots. 
Il existe nécessairement un entier $p \geq 1$ tel que 
$|y_p|<|z|$ et on finit par avoir obligatoirement $y_p=y'_p$.

De $y_{p-1}'=y_pz$ et $y_{p-1}=y_p'z=y_pz$, on tire que $y_{p-1}=y'_{p-1}$. Et on "remonte".
\end{proof}

\begin{remark}

Le lemme~\ref{LemEqua} peut s'obtenir aussi comme une conséquence de propriétés de la proposition~\ref{Lothaire}.
En effet, de $y'z=zy$, on obtient l'existence de deux mots $r$ et
$s$ sur $A$, et un entier $n$ tels que $z=r(sr)^{n}$, $y'=rs$
et $y=sr$.

On obtient alors que $yz(=z'y')$ finit par $y'(=rs)$ et  soit par 
$sr$ si $n \geq 1$ soit par $rr$ si $n=0$. Dans le premier cas, on obtient directement $y=y'$.

Dans le deuxième cas,  $yz=s(rr)=(z'r)s =z'y'$. Toujours en utilisant la propriété~\ref{Lotcase1} de la proposition~\ref{Lothaire},  il existe de deux mots $R$ et
$S$ sur $A$, et un entier $N$ tels que $s=R(SR)^{N}$, $z'r=RS$
et $rr=SR$. De plus, on a $|z'|=|r|$.

Si $|R|=|S|$ alors $r=R=S=z'$.

Si $|R|<|S|$ alors $|r|=\frac{1}{2}|SR|<|S|$ et $r$, suffixe de $S$, est facteur interne de $rr$. 
D'après le lemme~\ref{factint}, il existe
donc un mot non vide $t$ et des entiers non nuls $\alpha$ et $\beta$ tels que $r=t^{\alpha+\beta}$, 
$R=t^\beta$ et $S=t^{2\alpha+\beta}$.
Ce qui implique que $z'=r$.


Si $|R|>|S|$ alors $|r|>|S|$ et il existe deux mots non vides $r_1$ et $r_2$ tels que 
$r=r_1S=Sr_2$. En particulier, on a $|r_1|=|r_2|$.
Puisque $R$ finit par $r$ et par $r_1$, on en déduit que $r_1=r_2$ et que $z'=r_2r_1=r$.

Maintenant, puisque $yz=srr=rrs =z'y'$,  d'après la propriété~\ref{Lotcase3} de la proposition~\ref{Lothaire},
il existe deux mots $t_{1}$,  $t_{2}$ sur $A$
et trois nombres entiers $i$,  $j$,  $k$ tels que $s=(t_{1}t_{2}) ^ {i}t_{1}$,
$r=(t_{2}t_{1}) ^ {j}t_{2}$ et $r=(t_{1}t_{2}) ^ {k}t_{1}$.
On obtient soit que $r=t_{2}={t_1}$ si $j=k=0$ soit, puisque $r$ commence par $t_{2}t_{1}$ et par $t_{1}t_{2}$,
que ces deux derniers mots sont égaux sinon. Dans les deux cas,
$y=sr=(t_{1}t_{2}) ^ {i+j+1}=(t_{2}t_{1}) ^ {i+j+1}=rs=y'$.

\end{remark}

Par image miroir, on obtient directement~:

\begin{lemma}\label{LemEqua2} 

Soient $y$, $y'$, $z$ et $z'$ quatre mots.

Si $yz=zy'$ et $zy=y'z'$ alors $y=y'$.

\end{lemma}

\begin{lemma}\label{Equa3dontdemi}

Soient $x$, $x_1$, $x_2$ et $y$ quatre mots tels que $x$, $x_1$ et $x_2$ soient non vides et satisfaisant 
aux équations $x=x_1x_2$ et $x_2yx=yxx_1$ alors
il existe un mot $t$ et deux entiers $\alpha \geq 2$ et $\beta \geq 0$ tels que $x=t^\alpha$ et $y=t^\beta$.

\end{lemma}

\begin{remark}

On remarquera que les hypothèses du lemme~\ref{Equa3dontdemi} impliquent simplement que $yx$ 
est facteur interne de $(yx)^2$.

\end{remark}

\begin{proof}

D'après la propriété~\ref{Lotcase1} de la proposition~\ref{Lothaire},  il existe deux mots $u$ et
$v$ et un entier $q$ tels que $x_2=uv$, $yx=(uv)^qu$ et $x_1=vu$.
Puisque $x=x_1x_2=vuuv$, on a $q \geq 1$ et $yx$ finit par $vu$.
On en déduit que $vu=uv$.
D'après la propriété~\ref{Lotcase2} de la proposition~\ref{Lothaire},  il existe un mot $t$
et deux entiers $n$ et $p$ tels que $u=t^n$ et $v=t^p$.
On en déduit que $x=t^{2n+2p}$ et $y=t^{(q-2)(n+p)+n}$.
De plus, puisque $x_2=uv \neq \varepsilon$, on en déduit que $x$ n'est pas primitif.
\end{proof}

\begin{remark}

La composition de deux mots primitifs n'est pas nécessairement un mot primitif. Par exemple, si $u=aba$
et $v=bab$ alors $uv=(ab)^3$.

\end{remark}

\begin{proposition}\label{CompConjPrim}

Soit $w$ un mot primitif, soit $\underline{w}$ l'un de ses $|w|-1$ conjugués propres et 
soient $i \geq 1$ et $j \geq 1$ deux entiers alors $w^i\underline{w}^j$ est un mot primitif.

\end{proposition}

\begin{proof}

Soient $r$ et $s$ les mots tels que $w=rs$ et $\underline{w} =sr$.
Par définition de $\underline{w}$, on a $s \neq \varepsilon$ et $r \neq \varepsilon$.

Par l'absurde, supposons que $w^i\underline{w}^i=v^n$ pour un entier $n\geq 2$ et un mot primitif $v$.
Nous allons montrer que la majorité des cas amènent au fait que l'un des mots $v$ ou $w$ n'est pas 
primitif~:~ce qui nous conduit dans chacun de ces cas à une contradiction avec les hypothèses.

\quad \textbf{Cas 1 : $i = j = 1$}

\vsu

\neto Si $n=2$ alors $rs = sr$ et, d'après la remarque~\ref{RemLot}, $w$ n'est pas primitif.
Et, si $n \geq 4$ est pair, alors $w=v^{n/2}$ \cad\ à nouveau $w$ non primitif.

\vsu

\neto Si $n=2k+1$ est impair (avec $k\geq 1$), alors il existe deux mots $v_1$ et $v_2$ de même longueur 
($>0$) tels que $v=v_1v_2$ (en particulier, $|v|$ est paire) et $w=rs = v^kv_1$ et $\underline{w}=sr=v_2v^k$.

Soit $ 0 \leq \ell \leq k$ l'entier tel que $|v^\ell| \leq |r| <  |v^{\ell+1}|$.
Il existe un préfixe $v_1'$ de $v$ et un suffixe $v_2''$ de $v$ tels que $r=v^\ell v_1'=v_2''v^\ell$.
Soient $v_2'$ et $v_1''$ les mots non vides tels que $v=v_1'v_2'=v_1''v_2''$.

$\diamond$ Si $v_1' = \varepsilon$ alors $s$ commence par $v_1$ et par $v_2$. 
Cela implique que $v_1=v_2$ et donc que $v$ n'est pas primitif.

$\diamond$ Si $v_1' \neq \varepsilon$ et $\ell \geq 1$ alors, de $v^\ell v_1'=v_2''v^\ell$, on tire que
 $v^{\ell}$ est facteur interne de $(v^{\ell})^2$. 
Ce qui implique que $v$ est facteur interne de $v^2$. D'après le lemme~\ref{factint} et la
remarque~\ref{rem110}, $v$ n'est pas primitif.
 
$\diamond$ Si $v_1' \neq \varepsilon$ et $\ell =0$ alors $s=v_2'v^{k-1}v_1=v_2v^{k-1}v_1''$.

- Si $|v_2'|=|v_2|$ alors $|v_1|=|v_1''|$. Ce qui implique que $v_1=v_1'=v_1''$ et $v_2=v_2'=v_2''$.
Et puisque $r=v_1'=v_2''$, on tire que $v_1=v_2$ \cad\ $v$ non primitif.

- Si $|v_2'| \neq |v_2|$ et $k \geq 2$ alors $v$ est facteur interne de $v^2$. D'après le lemme~\ref{factint}
 et la remarque~\ref{rem110}, $v$ n'est pas primitif.

- Si  $|v_2'| > |v_2|=\frac{1}{2}|v|$ et $k =1$ (\cad\ $n=3)$ alors on a aussi $|v_1''|>|v_1|$. Ce qui signifie que 
 $v_2$ est un suffixe propre de $v_2'$ et que $v_1$ est un préfixe propre de $v_1''$. 
Puisque $rs=v_1v_2v_1$ avec $|r|=|v_1'|<|v_1|$ et puisque $s$ commence par $v_2v_1$,
on en déduit que $v_2v_1$ est facteur interne de $(v_2v_1)^2$.
D'après le lemme~\ref{factint}, $v_2v_1$ n'est pas primitif. Et,
d'après le corollaire~\ref{Corprim112}, $v=v_1v_2$ n'est pas primitif.

- Si  $|v_2'| < |v_2|=\frac{1}{2}|v|$ et $k =1$ alors $v=rv_2'=v_1''r$ avec $|r|=|v_1'|>|v_2'|$.
D'après la propriété~\ref{Lotcase1} de la proposition~\ref{Lothaire},  il existe de deux mots $R$ et
$S$ et un entier $q \geq 1$ tels que $r=R(SR)^{q}$, $v_1''=RS$
et $v_2'=SR$.  D'où $v=R(SR)^{q+1}$.
En particulier, si $R=\varepsilon$ ou si $S=\varepsilon$, alors $v$ n'est pas primitif.

Si $q$ est impair, soient $\rho_1$ et $\rho_2$ les mots de même longueur tels que $R=\rho_1\rho_2$,
$v_1=(RS)^{(q+1)/2} \rho_1$ et $v_2=\rho_2 (SR)^{(q+1)/2}$.
Puisque $s$ commence par $v_2' \rho_1=SR\rho_1$ préfixe de $v_2'v_1$ et par $\rho_2SR$ préfixe de $v_2$, 
on en déduit que $SR\rho_1= \rho_2SR$ . D'après le lemme~\ref{Equa3dontdemi},
$R$ et $S$ sont des puissances du même mot. 
Ce qui signifie que $v$ n'est pas primitf.

Si $q$ est pair, soient $\sigma_1$ et $\sigma_2$ les mots de même longueur tels que $S=\sigma_1\sigma_2$,
$v_1=(RS)^{q/2}R \sigma_1$ et $v_2=\sigma_2 (RS)^{q/2}R$.
Puisque $s$ finit par $v_1''=R \sigma_1\sigma_2$ et par $\sigma_2 R \sigma_1$
suffixe de $v_1$, on en déduit que $\sigma_1=\sigma_2$ et que $\sigma_1 R = R \sigma_1$.
D'après la propriété~\ref{Lotcase2} de la proposition~\ref{Lothaire},  $R$ et $\sigma_1$,  et donc 
$S=\sigma_1^2$, sont puissances du même mot.
Ce qui signifie à nouveau que $v$ n'est pas primitif.

\quad \textbf{Cas 2 : $i > j  \geq 1$}

Le mot $w^i$ est donc un préfixe commun d'une puissance de $w$ et d'une puissance de $v$.

Si $|w^i| \geq |v|+|w|$, d'après la proposition~\ref{Fine}, $w$ et $v$ sont des puissances du même mot primitif. 
On obtient donc $w=v$. Puisque $v^n$ finit par $\underline{w}^j$, il s'en suit que $w=\underline{w}$~:~une contradiction.

Si $|w^i| < |v|+|w|$ alors $|\underline{w}^j|=|w^j| < |v|$ et 
$|w^i|=|v|^{n}-|\underline{w}^j|>|v|^{n-1}>|v|^{n-2}+|w|$.
On a donc nécessairement $n=2$.
Il existe alors deux mots $w_1$ et $w_2$ non vides tels que
$w=w_1w_2$ et $v=w^{i-1}w_1=w_2 \underline{w}^j$.
Cela implique que $i-1=j$ et que $|w_2|=|w_1|$.
Il s'en suit que $w_2=w_1$ et que $w$ n'est pas primitif.

\quad \textbf{Cas 3 : $j > i  \geq 1$}

Ce cas se résout de la même façon que le précédent.

\quad \textbf{Cas 4 : $i = j \geq 2$}

Si $n$ est pair, on obtient $rs =sr$ et $w$ non primitif.

On a donc $n=2k+1$ impair avec $k \geq 1$.
Soient $v_1$ et $v_2$ les mots de même longueur tels que $v=v_1v_2$,
$w^i=v^{k}v_1$ et $\underline{w}^i=v_2v^{k}=v_2(v_1v_2)^{k}$.

Si $k \geq 2$ alors $|\underline{w}^i|=|w^i|=(k+\frac{1}{2})|v|$.
On a donc $|w^i| \geq 2|w|$ et $|w^i| \geq 2|v|$ \cad\ 
$|\underline{w}^{i}|=|w^{i}| \geq |v|+|w|=|v|+|\underline{w}|$. Inéquation que l'on obtient aussi
si $k=1$ et $i \geq 3$ puisque $|w^i|=|w|+\frac{i-1}{i}(k+\frac{1}{2})|v|$. 
Dans ces deux cas, cela implique, d'après la proposition~\ref{Fine} et le 
corollaire~\ref{CorFine1},  que $\underline{w}$, $w$ et $v$ sont des puissances du même mot primitif. 
On obtient donc $w=v=\underline{w}$. D'après le corollaire~\ref{NbreConjPrim} et la remarque~\ref{rem112},
on aurait $w$ non primitif~:~une contradiction.

Si $k=1$ et $i=2$, alors $w^2=rsrs=v_1v_2v_1$ commence par $v_1\underline{w}=v_1sr$ puisque $|v_1|=\frac{1}{2}|v|<|rs|$.
Si $v_1 \neq r$ alors  $sr$ est facteur interne de $(sr)^2$.
D'après le lemme~\ref{factint} et la remarque~\ref{rem110}, 
cela implique que $\underline{w}$ n'est pas primitif~:~une contradiction.
Si $v_1=r$ alors, puisque $\underline{w}^2=srsr$ finit par $v_2$ avec $|v_2|=|v_1|$,
on obtient que $v_2=r=v_1$ et donc que $v$ n'est pas primitif~:~une dernière contradiction.
\end{proof}

\begin{cor} \label{CorPrimPuis}

Soient $r$, $s$ et $z$ des mots non vides et soient $i \geq 1$, $j \geq 1$ et $n \geq 2$ trois entiers.
Si $(rs)^i(sr)^j = z^n$ alors $r$, $s$ et $z$ ont la même racine primitive.

\end{cor}

\begin{proof}

D'après la proposition~\ref{propconjracprim}, les mots $rs$ et $sr$ ont des racines primitives conjuguées.
D'après la proposition~\ref{CompConjPrim}, celles-ci sont égales.
Il existe donc un mot primitif $t$ et un entier $\alpha \geq 2$ tel que $rs=sr=t^{\alpha}$.
De plus, d'après la propriété~\ref{Lotcase1} de la proposition~\ref{Lothaire}, 
on obtient que $r$ et $s$ sont puissances du même mot qui ne peut être que $t$.
Et, d'après le corollaire~\ref{CorFine2}, $z$ est aussi une puissance de $t$.
\end{proof}

\begin{proposition} {\rm \cite{Lyn1962,Lot1983,HN2004,DH2006,Sha2008}} \label{PrimPuisXYZ}

Soient $x$, $y$ et $z$ des mots non vides et soient $m$, $n$ et $q $ des entiers supérieurs ou égaux à $2$.
Si $x^my^n = z^q$ alors $x$, $y$ et $z$ ont même racine primitive.

\end{proposition}

La proposition~\ref{PrimPuisXYZ} n'est qu'un corollaire de la proposition suivante :

\begin{proposition}

Soient $x$, $y$ et $z$ des mots primitifs et soient $m$, $n$ et $q $ des entiers supérieurs ou égaux à $2$.
Si $x^my^n = z^q$ alors $x=y=z$.

\end{proposition}

\begin{proof}

Puisque $x^m$ est préfixe de $z^q$, si $|x^{m-1}| \geq |z|$ alors, d'après la proposition~\ref{Fine},  $x$ et $z$ sont puissances du même mot primitif.
Par voie de conséquence, il en est de même de $y$. Et on obtient donc $x=y=z$.
De façon identique, si $|y^{n-1}| \geq |z|$, d'après le corollaire~\ref{CorFine1}, $x$, $y$ et $z$ sont puissances 
du même mot primitif et encore une fois $x=y=z$.

On a donc $(|x| \leq) \, |x^{m-1}| < |z|$ et $(|y| \leq) \, |y^{n-1}| < |z|$ \cad\ $|x^{m-1}| +|y^{n-1}| < 2|z|$ soit 
$|z^q|=|x^{m}| +|y^{n}| < 2|z|+|x|+|y|<4|z|$. Ce qui implique $q=2$ ou $3$.

Nous allons montrer que ces deux cas amènent à des contradictions avec les hypothèses
sur la primitivité de $x$, $y$ ou $z$.

Quitte à utiliser l'image miroir, sans perte de généralité, on peut supposer $|x| \geq |y|$.

\textbf{Cas 1 : $q=3$}

On a $|y^n| \leq |x^m|<2|z|$ donc il existe quatre mots non vides $x_1$, $x_2$, $y_1$ et $y_2$ tels que 
$z=x^{m-1} x_1=x_2y_1=y_2y^{n-1}$ avec $x=x_1x_2$ et $y=y_1y_2$.
Puisque $|z|=|x_2|+|y_1|<|x|+|y| \leq 2|x|$, on a nécessairement $m=2$.
Soit $y_1'$ le préfixe de $y_1$ tel que $x_1x_2=x_2y_1'$.
D'après la propriété~\ref{Lotcase1} de la proposition~\ref{Lothaire}, il existe deux mots $r$ et
$s$ et un entier $\ell$ tels que $x_2=(rs)^{\ell}r$, $x_1=rs$
et $y_1'=sr$. 

Il s'en suit que $x=(rs)^{\ell+1}r$ et $z=xx_1=(rs)^{\ell+1}rrs=x_2y_1=(rs)^{\ell}ry_1$
et que $y_1=srrs$. 

Si $r=\varepsilon$ ou si $s=\varepsilon$, on obtient que 
$x,y$ et $z$ ne sont pas primitifs~:~contraire aux hypothèses.

Puisque $|x| \geq |y|>|y_1|$, on a $\ell \geq 1$.
De plus, puisque $|x_2|+|y_1|=|z|= |y_2y^{n-1} |\geq |y_1|+2 |y_2|$, on obtient donc que
$2|y_2| \leq |x_2|$ et que $y_2$ est un préfixe de $x_2$.
Il existe donc un entier $0 \leq q \leq \ell$ et 
deux mots $t_1 \neq rs$ et $t_2 \neq \varepsilon$ tels que
$t_1t_2=rs$ et $y_2=(rs)^{q}t_1$ \cad\ $x_2=y_2t_2(rs)^{\ell -q -1}r$.
En outre, $z$ commence par $y_2y_1=(rs)^{q}t_1srrs$ 
et par $x=(rs)^{\ell+1}r$. 


Puisque $2|y_2| \leq |x_2|$,
on a $q< \ell$ et 
$z$ commence par $(rs)^{q}t_1srrs$ et par $(rs)^{q}rsrsr$.
Cela implique que $t_1sr=rst_1$.
D'après la propriété~\ref{Lotcase3} de la proposition~\ref{Lothaire}, 
il existe deux mots $u_{1}$,  $u_{2}$ sur $A$
et trois nombres entiers $\alpha$,  $\beta$,  $\gamma$ tels que $t_1=(u_{1}u_{2}) ^ {\alpha}u_{1}$,
$s=(u_{2}u_{1})^{\beta}u_{2}$ et $r=(u_{1}u_{2})^{\gamma}u_{1}$.
Puisque $z$ finit par $rs$ donc par $u_1u_2$  et par $y_1y_2$ donc par $u_2u_1$,
on en déduit que $u_1u_2=u_2u_1$.
D'après la propriété~\ref{Lotcase2} de la proposition~\ref{Lothaire}, 
les mots $u_1$ et $u_2$ sont des puissances du même mot.
Il en est donc de même de $r$ et $s$.
Il s'en suit que $x,y$ et $z$ ne sont pas primitifs~:~contraire aux hypothèses.

\textbf{Cas 2 : $q=2$}

Puisque $z$ est primitif, pour tous les entiers $2 \leq m' \leq m$ et $2 \leq n' \leq n$,
on a $|x^{m'}| \neq |z|$, $|y^{n'}| \neq |z|$.

\neto Si $|x^m|<|z|$, il existe deux mots non vides $y_1$ et $y_2$ tels que 
$z=x^m y_1=(y_2y_1)^{n-1}y_2$ avec
$y=y_1y_2$. Puisque $|x^m| \geq |x^{m-1}| + |y| \geq |x| + |y_2y_1|$, d'après la proposition~\ref{Fine},  
$x$ et $y_2y_1$ sont puissances du même mot. Puisque $x$ et $y$ sont primitifs et d'après le
corollaire~\ref{Corprim112}, on obtient $x=y_2y_1$ et $z=(y_2y_1)^my_1=(y_2y_1)^{n-1}y_2$.
Ce qui implique $m=n-1$, $y_1=y_2$ et $y$ non primitif~:~une contradiction.

\neto Si $|x^m|>|z|$, il existe deux mots non vides $x_1$ et $x_2$ tels que $z=(x_1x_2)^{m-1}x_1=x_2y^{n}$ avec 
$x=x_1x_2$. 

Si $|y^n| \geq |x| + |y|$ alors, comme dans le cas précédent, on obtient $y=x_2x_1$ puis $x_1=x_2$ et $x$ 
non primitif~:~une contradiction. 

On a donc $|y^n| < |x| + |y|$.
Puisque $|z|=|x_2y^n|\leq |x_2|+2|x|$, on obtient que $m=2$ ou $m=3$.
Soit $Y$ le préfixe de $y^{n}$ tel que $x_1x_2=x_2Y$.
D'après la propriété~\ref{Lotcase1} de la proposition~\ref{Lothaire}, il existe deux mots $r$ et
$s$ et un entier $\ell$ tels que $x_2=(rs)^{\ell}r$, $x_1=rs$
et $Y=sr$. En particulier, on obtient $x=(rs)^{\ell+1}r$.

On a  $r \neq \varepsilon$ et $s \neq \varepsilon$ car sinon $x$ ne serait pas primitif.

$\diamond$ Si $m=2$ alors, puisque $z=xx_1=(rs)^{\ell}rsrrs=x_2y^n$, 
on obtient que $y^n=srrs$. D'après le corollaire~\ref{CorPrimPuis},
cela implique que $x$ n'est pas primitif~:~une contradiction.

$\diamond$ Si $m=3$ alors, puisque $z=x^2x_1=x_2y^n$, 
on obtient que $y^n=sr \, (rs)^{\ell+1}r \, rs =sr \, x \, rs$. 
Puisque $|y^n| < |x| + |y|$, il s'en suit que $|y| > 2|r|+2|s|$.
Mais $|(rs)^{\ell+1}r|=|x|=|y^n|+2|r|+2|s| \geq |y|$.
Cela implique que $\ell \geq 1$.
%

- Si $n \geq 3$ alors $x$ est un facteur commun d'une puissance de $rs$ et d'une puissance de $y$
avec $|x|=|(rs)^{\ell+1}r|\geq 2|rs|$ et $|x|>|y^{n-1}| \geq 2|y|$ \cad\ $|x| \geq |rs| +|y|$. 
D'après le corollaire~\ref{Kera},  $y$ n'est pas primitif car $|y|>|rs|$.

%
%

- Si $n=2$ et $\ell$ est impair, soient $\rho_1$ et $\rho_2$ les mots de même longueur tels que 
$r=\rho_1\rho_2$ et $y=sr (rs)^{(\ell+1)/2}\rho_1=\rho_2(sr)^{(\ell+1)/2} rs$.
On en déduit que $sr \rho_1= \rho_2 sr$.
D'après le lemme~\ref{Equa3dontdemi}, $r$ et $s$ sont puissances du même mot.
Cela signifie que $x$ n'est pas primitif~:~une contradiction.

- Si $n=2$ et $\ell$ est pair, comme dans le cas précédent, avec $\sigma_1$ et $\sigma_2$ les mots de 
même longueur tels que $s=\sigma_1\sigma_2$, on trouve de même que 
$y=sr (rs)^{\ell/2}r\sigma_1=\sigma_2r(sr)^{\ell/2} rs$.
On en déduit que $\sigma_1=\sigma_2$ et $\sigma_1r=r\sigma_1$.
D'après la propriété~\ref{Lotcase2} de la proposition~\ref{Lothaire},  $r$ et $\sigma_1$,  et donc 
$s=\sigma_1^2$, sont puissances du même mot.
Ce qui signifie que $y$ n'est pas primitif~:~une dernière contradiction.
\end{proof}

\begin{cor} \label{CorLot4}

Si $x$ et $y$ sont deux mots primitifs différents alors $x^my^n$ est un mot primitif pour tous les entiers $m$ et $n$ plus grands que $2$.

\end{cor}

\section{Caractérisation des morphismes primitifs}

\begin{proposition}\label{PrimInj}{\rm{(Theorem 5 dans~\cite{Mit1997})}}

Un morphisme primitif est injectif.

\end{proposition}

\begin{proof}

Soit $h$ un morphisme  primitif de $A^*$ vers $B^*$

Remarquons dans un premier temps que $h$ ne peut pas être un morphisme effaçant. En effet, si c'était le cas, 
on pourrait trouver $\alpha,\beta\in A$ tels que $h(\alpha)=\varepsilon$ et $h(\beta)\neq\varepsilon$. 
On aurait bien que $\beta \alpha \beta$ est un mot primitif mais ce n'est pas le cas de son image par $h$ 
car $h(\beta \alpha \beta)=(h(\beta))^2$.

Soient $x$ et $y$ deux mots de $A^+$ tels que $h(x)=h(y)$. Sans perte de généralité, on suppose que $|x| \leq |y|$.

Puisque $h$ est primitif et puisque $h(xy)=h(x)^2$ et $h(xyy)=h(x)^3$,
les deux mots $xy$ et $xyy$ ne sont pas primitifs. 

Autrement dit, on peut écrire $xy=u^n$ et $xyy=v^m$ pour des mots primitifs non vides $u$ et $v$ et deux entiers $n,m \geq 2$.

Si $m=2$ alors, puisque $|x| \leq |y|$, on peut trouver deux mots non vides
$y_1$ et $y_2$ tels que $y=y_1y_2$ et $v=xy_1=y_2y$. Cela signifie que $x$ est un facteur propre ($\neq y$) de $y$.
Puisque $h$ est non effaçant, on obtient $|h(x)| <|h(y)|$. Cette situation est donc impossible.

Si $n=2$ alors, à nouveau puisque $x$ ne peut pas être facteur propre de $y$, on obtient $x=y$ et l'injection est montrée.

On a donc $n,m \geq 3$. En particulier, on obtient que $|u|\leq \dfrac{1}{3}(|x|+|y|)$ et $|v|\leq \dfrac{1}{3}(|x|+2|y|)$.
Et donc $|u| + |v| \leq \dfrac{2}{3}|x|+|y|  <|x|+|y|$.
Cela signifie que $u^n$ et $v^m$ ont un préfixe commun  de longueur supérieure ou égale à $|u| + |v|$.
D'après la proposition~\ref{Fine}, puisque $u$ et $v$ sont primitifs, on obtient $u=v$.

Il s'en suit que $y=u^{m-n}$ et $x=u^{2n-m}$. De l'égalité,
$h(x)=h(y)$, on tire $2n-m=m-n$ et $x=y$.
\end{proof}

\begin{remark}
On trouve une autre démonstration de la proposition~\ref{PrimInj} dans~\cite{Hsi2003} (Proposition 5.2).
Elle utilise le corollaire~\ref{CorLot4} lorsque $h(x)=h(y)$ et  lorsque $x$ et $y$ n'ont pas la même racine primitive.
Dans ce cas, $x^2y^2$ est un mot primitif dont l'image n'est pas primitive.

\end{remark}

\begin{remark}\label{RemPrimInj}

Un morphisme uniforme 2-primitif est injectif.

En effet, si $h$ est uniforme, avoir $h(x)=h(y)$ avec $x\neq y$, signifie qu'il existe une lettre $x_i$ de $x$ 
et un lettre $y_i$ de $y$ telles que $x_i \neq y_i$ et $h(x_i)=h(y_i)$. On a donc $x_iy_i$ primitif et pas $h(x_iy_i)$.

\end{remark}

\begin{proposition} {\rm{(Theorem 5 dans~\cite{Mit1997})}}

Un morphisme $h$ de $A^*$ vers $B^*$ est primitif si et seulement si $h(A)$ est un code pur.

\end{proposition}

\begin{proof}

\begin{description}

\item[$(\Leftarrow)$] Supposons que $h(A)$ est un code pur et soit $x$ un mot primitif de $A^+$.

Si $h(x)=u^k$ pour un mot primitif $u$ de $A^+$, on a $u^k \in h(A)^*$ et, puisque $h(A)$ est pur, on a  $u \in h(A)$.
Il existe donc un mot $y \in A^+$ tel que $u=h(y)$. Cela implique que $h(x)=h(y^k)$. Puisque $h$ est injective
($h(A)$ est un code), on a donc $x=y^k$.
Enfin, $x$ étant primitif, on obtient $k=1$ et $h(x)$ primitif.

\item[$(\Rightarrow)$] Supposons que $h$ est primitif.

Soit $w=h(v)$ un mot de $h(A)^*$ avec $r=\rho(v)$  $s=\rho(w)$. 

Soient $n,m \geq 1$ les entiers tels que $v=r^n$ et $w=s^m$.

Puisque $s$ et $h(r)$ sont primitifs, d'après le corollaire~\ref{CorFine2}, de l'égalité $s^m=h(r)^n$, 
on obtient bien que $s=h(r) \in h(A)$.

\end{description}

\end{proof}

\begin{lemma}\label{LemEqua3}

Soit $f$ un morphisme uniforme de $\{a,b\}^*$ dans $B^*$.

Soient $\alpha,\beta,\gamma$ et $\delta$ des lettres de $\{a,b\}$ avec $\alpha \neq \beta$.

Si $f(\alpha)f(\beta)=Xf(a)Y$ ou si $f(\alpha)f(\beta)=Xf(b)Y$ avec $X$ un suffixe non vide de $f(\gamma)$
et $Y$ un préfixe non vide de $f(\delta)$ alors $f$ n'est pas primitif.

Plus précisemment, il existe un mot primitif de longueur inférieure ou égale à deux dont l'image n'est pas
primitive.

\end{lemma}

\begin{proof}

Nous allons montrer que l'un des mots $f(a)$, $f(b)$, $f(ab)$  ou $f(ba)$ n'est pas primitif.

Par symétrie, nous ne traitons que l'équation $f(\alpha)f(\beta)=Xf(b)Y$.

Remarquons que, si on avait pu avoir $\alpha=\beta=b$ alors $f(b)$ aurait été un facteur interne de $f(b)f(b)$.
D'après le lemme~\ref{factint}, $f(b)$ n'aurait pas été primitif.

Par image miroir, sans perte de généralité, on peut supposer $\alpha=a$ et $\beta=b$.

Si $\delta=b$ alors $f(b)$ est facteur interne de $f(b)f(b)$.
D'après le lemme~\ref{factint}, $f(b)$ n'est pas primitif.
On suppose donc $\delta=a$.

Si $\gamma=a$ alors $f(ab)$ est facteur interne de $f(ab)f(ab)$.
D'après le lemme~\ref{factint}, $f(ab)$ n'est pas primitif.
On suppose donc $\gamma=b$.

\neto Si $|X| = |Y|$, on a $2|X|=|XY|=|f(b)|$. Puisque $f(b)$ finit par $X$ et par $Y$, on obtient $X=Y$.
Et puisque $f(b)Y$ finit par $f(b)$, on a $f(b)=X^2$ : $f(b)$ n'est pas primitif.

\neto Si $|X| > |Y|$, puisque $f(a)$ commence par $Y$ et par $X$, il existe un mot non vide $X'$ 
tel que $X=YX'$. 
Puisque $|XY|=|f(a)|=|f(b)|$ et que $f(b)Y$ finit par $XY=YX'Y$ et par $f(b)$ et donc par $X$,
on obtient que $X'Y=YX'$.
D'après la propriété~\ref{Lotcase2} de la proposition~\ref{Lothaire}, cela implique que 
$X'$ et $Y$ sont des puissances du même mot. Cela signifie que $f(b)=YX'Y$ n'est pas primitif.

\neto Si $|X| < |Y|$, puisque $f(b)$ finit par $X$ et par $Y$, il existe un mot non vide $Y'$ tel que $Y=Y'X$. 
Puisque $|XY|=|f(a)|=|f(b)|$ et que $f(b)Y$ finit par $XY$ et par $f(b)$,
on obtient que $f(b)=XY'X$.

Il s'en suit que $f(a)=(XX)Y'$ et $f(a)$ commence par $Y=Y'X$. Il existe donc un mot $Z$
tel que $f(a)=Y'(XZ)$ avec $|Z|=|X|$.

En prenant, $y=XY'$, $y'=Y'X$, $z=X$ et $z'=Z$, d'après le lemme \ref{LemEqua2}, on obtient que
$XY'=Y'X$. D'après la propriété~\ref{Lotcase2} de la proposition~\ref{Lothaire}, cela implique que 
$X$ et $Y'$ sont des puissances du même mot. Cela signifie que $f(a)$ n'est pas primitif.
%
%
%
%
%
%
\end{proof}

\begin{proposition} {\rm{(Theorem 12 dans~\cite{Mit1997})}}

Soit $n\geq 2$ un entier. Il existe des morphismes binaires qui sont primitifs jusqu'à $n$ mais qui ne sont pas primitifs.

\end{proposition}

\begin{proof}

On considère le morphisme $f$ de $\{a,b\}^*$ vers $\{a,b\}^*$ défini par
$f(a)=aba$ et $f(b)=(baa)^{n-1}b$.

On a $f(a^nb)=(aba)^{n-1}aba \, (baa)^{n-1}b=(aba)^{n-1}ab \mid (aba)^{n-1}ab$.

Remarquons dans un premier temps que $f(a)$ commence et finit par $a$ et que
$f(b)$ commence et finit par $b$.
Ce qui implique que $f$ est injective.

\textit{Fait 1 :} Si $f(b)=sf(a)p$ avec $s$ un suffixe de l'image d'un mot et $p$ le préfixe de l'image d'un mot,
alors, puisque $f(b)=ba(aba)^{n-2}ab$,  on a nécessairement $s=ba(aba)^{\alpha}$ et
$p=(aba)^{\beta}ab$ avec $\alpha + \beta =n-3$. Ce qui signifie que $s$ est un suffixe de $f(a^{\alpha+1})$
et $p$ est un préfixe de $f(a^{\beta+1})$.

\textit{Fait 2 :} On ne peut pas avoir $sf(b)=f(b)p$ avec $s$ un suffixe propre ($\neq f(b)$) de l'image d'un 
mot et $p$ le préfixe propre de l'image d'un mot.
En effet, l'équation $sf(b)=sb(aab)^{n-2}aab=b(aab)^{n-2}aab \, p$ implique qu'au moins deux 
facteurs $aab$ de chacun des mots de l'équation soient alignés.
Cela signifierait que $p$ commencerait par $(aab)^{\ell}aab$ pour un entier $\ell$ : c'est impossible.

\vsu
Si $f(w)=u^k$ pour un entier $k\geq 2$ avec $w$ primitif alors $w$ contient au moins une fois la lettre $b$.
Si l'une des deux occurences $u^{k-1}$ contient $f(b)$, d'après les faits 1 et 2, alors $|w| \geq 1 + n$.
Sinon, cela signifie que $|f(b)|>|u^{k-1}|$. Mais il existe bien un facteur $f(a)$ dans $u$. Toujours
d'après le fait 1, on obtient à nouveau  $|w| \geq 1 + n$.
\end{proof}

\begin{proposition} \label{Pro2prim} {\rm{(Theorem 10 dans~\cite{Mit1997})}}

Un morphisme uniforme binaire est primitif si et seulement s'il est 2-primitif.

\end{proposition}

\begin{proof}

L'implication étant naturelle, on s'intéressera uniquement à la réciproque que l'on montre par contraposée.

Soit $L \geq 1$ un entier et soit $f$ un morphisme $L$-uniforme  de $\{a,b\}^*$ vers $B^*$.

Soit $w$  un mot primitif de longueur $n$.
On suppose que $f(w)$ n'est pas primitif \cad\ qu'il existe un mot non vide $u$ et un entier $k \geq 2$ tels que 
$f(w)=u^k$. Quitte à considérer sa racine primitive, on peut supposer sans perte de généralité que
$u$ est primitif.
On suppose aussi que la longueur de $w$ (et par conséquent de $f(w)$) est minimale. Si $n \leq 2$, cela termine la preuve.
On travaillera donc par l'absurde avec $n \geq 3$.

D'après la remarque~\ref{RemPrimInj}, si $f$ n'est pas injectif alors $f$ n'est pas $2$-primitif. Ce
qui termine à nouveau la preuve.

On a $|f(w)|=|w| \times L= k \times |u|$. Puisque $f$ est injectif, $|u|$ ne peut pas être un multiple de $L$ 
(\cad\ $k$ un diviseur de $|w|$) sinon $w$ ne serait pas primitif.

\textit{Cas 1 : $|u|<L$}

Soit $i \geq 2$ le plus petit indice tel que $\sub{w}{i}=\sub{w}{1}$. 
Si un tel indice n'existait pas, il suffit alors de considérer l'image miroir de $f(w)$.

Puisque $f(\sub{w}{i})$ est facteur de $u^k$, il existe un suffixe $u_1$ de $u$, 
un préfixe $u_2$ de $u$ et un entier $j \geq 0$ tels que  $f(\sub{w}{i})=u_1u^ju_2$.

Si $u_1=\varepsilon$ alors $f(\sub{w}{1.i-1})$ est une puissance de $u$ : contraire à l'hypothèse de la longueur minimale de $w$.

Si $u_1 \neq \varepsilon$ alors $u$ préfixe de $f(\sub{w}{1})$ est facteur interne de $uu$. Ce qui signifie que 
$u$ n'est pas primitif~: contraire à l'hypothèse sur $u$.

\textit{Cas 2 : $|u| > L$}

\quad \textit{Cas 2.1 : $k=2$}

Cela signifie que $n=2p+1$ est impair. 

Il existe deux mots non vides $x$ et $y$  tels que $u=f(\sub{w}{1..p})x= yf(\sub{w}{p+2.. n})$,
$f(w_{p+1})=xy$ et $|x|=|y|$.

On en déduit que $f(\sub{w}{1})$ commence par $y$
et $f(\sub{w}{n})$ finit par $x$.

Si $x=y$ (ce qui est le cas lorsque  $\sub{w}{1}= \sub{w}{p+1}$ ou lorsque  
$\sub{w}{n} = \sub{w}{p+1}$) alors 
$f(\sub{w}{p+1})$ n'est pas primitif~: fin de la preuve.
On a donc $x \neq y$, $\sub{w}{1}=\sub{w}{n} \neq \sub{w}{p+1}$ et  $f(\sub{w}{1})=f(\sub{w}{n})=yx$.

On a donc soit $\sub{w}{1}=\sub{w}{n}=a$ et $\sub{w}{p+1}=b$ soit $\sub{w}{1}=\sub{w}{n}=b$ 
et $\sub{w}{p+1}=a$.
Ces deux cas étant symétriques, on ne traite que le premier.

Soit $\ell_1$ le plus petit entier tel que $\sub{w}{1+\ell_1}=b$
et 
soit $\ell_2$ le plus petit entier tel que $\sub{w}{p+1+\ell_2}=a$.
De tels entiers existent et on a $\ell_1 \leq p$ et $\ell_2 \leq p$.

Le mot $uy$ commence par $f(a^{\ell_1})f(b)=(yx)^{\ell_1}xy$ et par $y \,f(b)^{\ell_2-1}f(a)=(yx)^{\ell_2}yyx$.

Si $\ell_1 = \ell_2$, alors $x=y$ (et $f(b)$ n'est pas primitif)~: contraire au cas présent.

Si $\ell_1 \neq \ell_2$,  alors $xy=yx$ et $f(b)$ n'est pas primitif~: contraire aux hypothèses.

\quad \textit{Cas 2.1 : $k \geq 3$}

Soit $i_2$ le plus petit indice tel que $u$ soit préfixe de $f(\sub{w}{1..i_2})$
et soit $i_3$ le plus petit indice tel que $uu$ soit préfixe de $f(\sub{w}{1..i_3})$.

Il existe des mots $p_2$, $s_2$, $p_3$ et $s_3$ tels que $f(w_{i_2})=p_2s_2$
et $f(w_{i_3})=p_3s_3$. Par définition de $i_2$ et $i_3$, on ne peut pas avoir $p_2 = \varepsilon$
ou $p_3 = \varepsilon$. On exclut aussi les cas $s_2 = \varepsilon$
ou $s_3 =\varepsilon$ par la primitivité de $w$ ou par la minimalité de longueur de $w$.

Remarquons aussi qu'avoir $|p_2|=|p_3|$ (ou $|s_2|=|s_3|$) impliquerait que $|u|$ soit un multiple de $L$
et donc que $p_2 = \varepsilon$, c'est un cas que nous venons d'exclure.

Soit $X$ le mot tel que $Xp_2$ soit un suffixe de $u$ de longueur $L$
et soit $Y$ le mot tel que $s_2Y$ soit un préfixe de $u$ de longueur $L$. On a donc
$Xf(\sub{w}{i_2})Y=f(\sub{w}{1})f(\sub{w}{n})$.

Si $\sub{w}{1} \neq \sub{w}{n}$, d'après le lemme~\ref{LemEqua3}, il existe un mot primitif de longueur 
inférieure ou égale à deux dont l'image n'est pas primitive~: fin de la preuve.

Si $\sub{w}{1}=\sub{w}{n}=\sub{w}{i_2}$  alors $f(\sub{w}{i_2})$ est facteur interne de $f(\sub{w}{i_2})f(\sub{w}{i_2})$~: $f(\sub{w}{i_2})$ n'est pas primitif.

De même, si $\sub{w}{1}=\sub{w}{n}=\sub{w}{i_3}$, on obtient que $f(\sub{w}{i_3})$ est facteur interne de 
$f(\sub{w}{i_3})f(\sub{w}{i_3})$ et qu'il n'est donc pas primitif.

Il nous reste donc le cas $\sub{w}{1}=\sub{w}{n} \neq \sub{w}{i_2}=\sub{w}{i_3}$. 
Sans perte de généralité, on peut supposer les deux premiers égaux à $a$ et les deux suivants égaux à $b$.

Soit $\ell_1$ le plus grand entier tel que $\sub{w}{\ell_1}=a$.
Un tel entier existe et on a $1 \leq \ell_1 \leq i_2 -1$. 
On considère alors le facteur $f(\sub{w}{\ell_1})f(\sub{w}{\ell_1+1})=f(a)f(b)$ du premier $u$
et son occurence dans le deuxième $u$. 
D'après le lemme~\ref{LemEqua3}, il existe un mot primitif de longueur inférieure ou égale
à deux dont l'image n'est pas primitive.
\end{proof}

\section{Morphismes primitifs et puissance $k \geq 2$}

\begin{lemma}  {\rm{(Proposition 5.4 dans~\cite{Hsi2003})}}

Soit $f$ un morphisme injectif de $A^*$ vers $B^*$ et soit 
$w$ un mot primitif de $A^+$.
On suppose que $f(w)=u^m$ avec $u \neq \varepsilon$ primitif et $m\geq 1$ un entier.
Pour tout mot $v\in A^+$, on a  $f(v)=u^\ell$ avec $\ell \geq 1$ un entier si et seulement si $v$ est une puissance de $w$.

\end{lemma}

\begin{proof}

Puisque $f(v^m)=u^{ \ell m}=f(w^\ell)$ et que $f$ est injective, on en déduit que $v^m=w^\ell$.
D'après le corollaire~\ref{CorFine2}, on obtient que $v$ et $w$ sont puissances d'un même mot.
Ce mot ne peut être que $w$ lui-même puisque celui-ci est primitif.
\end{proof}

\vsd

La propriété suivante est enoncée pour les morphismes sans carré dans~\cite{Mit1996}. Mais elle est vraie pour tout entier $k \geq 2$.

\begin{lemma}\label{BifInj}

Soit $f$ un morphisme de $A^*$ vers $B^*$ et soit $k \geq 2$ un entier.

Si $f$ est sans-puissance $k$ alors $f$ est bifixe. C'est donc un morphisme injectif.

\end{lemma}

\begin{proof}

Par exemple, si $f$ n'était pas un morphisme préfixe, il existerait deux lettres différentes $x$ et $y$ tel que 
$f(x)$ serait préfixe de $f(y)$. 
Dans ce cas, l'image du mot $x^{k-1}y$ qui est sans puissance $k$ contiendrait $f(x)^{k}$.

De même, si $f$ n'était pas un morphisme suffixe.
\end{proof}

\begin{cor}

Un morphisme uniforme binaire sans puissance $k \geq 2$ est primitif.

\end{cor}

\begin{proof}

Par contraposition, on suppose qu'un morphisme $f$ uniforme défini sur $\{a,b\}$ n'est pas primitif.
D'après la proposition~\ref{Pro2prim}, on a $f(x)=u^m$  avec $x \in\{a,b,ab,ba\}$, $u$ un mot non vide 
et $m \geq 2$ un entier.

On a donc $f(x^{k-1})=u^{m(k-1)}$. Mais $x^{k-1} \in\{a^{k-1},b^{k-1},(ab),(ba)^{k-1}\}$
est un mot sans puissance $k$ et $m(k-1) \geq k$~: $f$ n'est pas sans puissance $k$.
\end{proof}

\vsd

Le lemme suivant est énoncé pour les morphismes sans carré dans~\cite{Mit1997} (Corollaire~7). Mais il est vrai (avec sensiblement la même démonstration) pour tout entier $k \geq 2$.

\begin{lemma}
 
Il existe des morphismes primitifs qui ne sont pas sans-puissance $k$ pour tout entier $k \geq 2$.

\end{lemma}

\begin{proof}

Soit $k \geq 2$ un entier. On considère le morphisme $f$ de $\{a,b,c\}^*$ vers $\{a,b,c\}^*$ défini par
$f(a)=ac^k$, 
$f(b)=bc^k$ et
$f(c)=abc^k$.

Supposons qu'il existe un mot $w$
tel que $f(w)$ ne soit pas primitif.
On peut écrire $f(w)=u^p$ pour un mot $u$ non vide et un entier $p\geq 2$.

Par un critère de longueur, le mot $u$ finit nécessairement par $c^k$ (sinon $u$ ne contiendrait que des $c$ ce qui est absurde).

Cela signifie qu'il existe des entiers $0 \leq i_q \leq |w|$ tels que $i_0=0$ et 
$u=f(\sub{w}{i_{q-1}+1..i_q})$ pour tout $1 \leq q \leq p$.
Puisque $f$ est un morphisme préfixe, on a $f$ injectif. Cela implique que tous les $\sub{w}{i_{q-1}+1..i_q}$
sont égaux et donc que 
$w=(\sub{w}{1..i_1})^p$~:~$w$ n'est pas primitif.
\end{proof}

\vst

On définit les entiers ${\left( t_k \right)}_{k \geq 2}$ par $t_2=3$, $t_3=4$,
$t_k=\frac{k^2}{2}$  si $k \geq 4$ est pair et
$t_k=\frac{k \times (k - 1)}{2} +2$ si $k \geq 5$ est pair.
En particulier, cela signifie que, quand $k \geq 4$, on a
$t_k=k \lfloor \frac{k}{2} \rfloor + 2(k \bmod 2)$. 

\vsu

\begin{prop}
\label{pur2}
\cite{Wla2001}
Soit $k \geq 2$ un entier.
Si un morphisme binaire est sans puissance $k$ jusqu'à $t_k$ alors il est primitif.

\end{prop}

\begin{remark}

La borne $t_k$ est optimale et améliore celle donnée par Leconte dans sa thèse~\cite{Lec1985}
qui est $\frac{k \times(k+1)}{2}$.

\end{remark}

\begin{cor}

Un morphisme binaire sans puissance $k$ est primitif.

\end{cor}

La preuve de la proposition~\ref{pur2} est basée sur un résultat de Lentin and Sch\"utzenberger :

\begin{lemma}{\rm \cite{LS1967}}
\label{Lentin}
Un morphisme $f$ sur $\{ a,b \}$ est primitif si et seulement si
$f(w)$ est primitif pour tous les mots $w \in a^*b \cup ab^*$.
\end{lemma}

\begin{lemma}\label{lemmepuissmax}

Soit $w$ un mot primitif et soit $j \geq 1$ un entier.

Soit $k_j$ le plus grand entier tel que $w^j$ contienne une puissance $k_j$.
On a $k_{j} \leq \max\{j,k_1,k_2,k_3\}$.

En particulier, il existe un entier $j_0 \geq 1$ tel que 
si $j \geq j_0 $ alors $k_{j} = j$.
\end{lemma}

\begin{remark}

Autrement dit, il existe un entier $j_0 \geq 1$ tel que 
si $j \geq j_0 $ alors $w^j$ est sans puissance $j+1$.

\end{remark}

\begin{proof}

Le résultat est trivial si $j=1$, $j=2$ ou $j=3$.

%
%
%
%
%
%
%
%
On va montrer la propriété par récurrence pour $q \geq 4$.

On suppose que $k_{q-1} \leq \max\{q-1,k_1,k_2,k_3\}$ pour un entier $q \geq 4$.


Si $k_{q}=k_{q-1}$, on obtient immédiatement que $k_{q} \leq \max\{q-1,k_1,k_2,k_3\} \leq \max\{q,k_1,k_2,k_3\}$.

Si $k_{q}>k_{q-1}$,
soit $v$ un mot non vide tel que $w^q=pv^{k_q}s$. Par définition de $k_q$, le mot $v$ est primitif. 
On a $|p|<|w|$ car sinon $v^{k_q}$ serait facteur de $w^{q-1}$; ce qui est en contradiction avec
l'hypothèse $k_{q}>k_{q-1}$. De même, on a nécessairement $|s|<|w|$.

Cela signifie que $|v^{k_q}|>|w^{q-2}| \geq 2|w|$.
De plus, on a $|v^{k_q}| \geq |v^{k_{q-1}}| + |v | \geq 2|v|$.
Le mot $v^{k_q}$ est donc un facteur commun d'une puissance de $v$ et d'une puissance de $w$
de longueur supérieure ou égale à $|v|+|w|$.

D'après le corollaire~\ref{Kera}, il existe deux mots $t_1$ et $t_2$ tels que $v$ soit une puissance de $t_1t_2$
et $w$ soit une puissance de $t_2t_1$. Puisque $v$ et $w$ sont primitifs, on en déduit que $v=t_1t_2$
et $w=t_2t_1$.
Si $t_1 \neq \varepsilon$, on obtient $k_q=q-1$ et si $t_1 =\varepsilon$, on obtient $k_q=q$.
\Cad\  $k_{q} \leq \max\{q,k_1,k_2,k_3\}$
\end{proof}

\begin{remark}

Le fait que $k_2$ ne peut être majoré par $\{2, k_1\}$ est mis en évidence en utilisant par exemple 
$w=ca(bc)^nb$ avec $n \geq 2$. Et
le fait que $k_3$ ne peut être majoré par $\{3, k_1, k_2\}$ est mis en évidence en utilisant par exemple 
$w=bcab(abcab)^nabca$ avec $n \geq 2$.

\end{remark}

\begin{proposition}\label{PuisskPrim}

Un morphisme sans puissance $k$ avec $k \geq 5$ est primitif.

\end{proposition}

\begin{proof}

Par contraposée, on suppose qu'un morphisme $f : A^* \rightarrow B^*$  n'est pas primitif.

Soit $w$ un mot primitif tel que $f(w)=u^n$ avec $n\geq 2$.

D'après la propriété~\ref{lemmepuissmax}, il existe un entier $j_0$ tel que pour tout $k \geq j_0$ la puissance
maximale dans $w^k$ soit inférieur ou égale à $k$.
Autrement dit, le mot $w^k$ est donc sans puissance $k+1$

Mais $f(w^k)=u^{n \times k}$ avec $n \times k \geq k+1$. Cela signifie que $f$
n'est pas sans puissance $k+1$ pour tout $k \geq \max\{5;j_0\}$.
D'après la propriété~\ref{puissancek}, $f$ n'est donc pas sans puissance $k$ pour tout entier $5 \leq k \leq \max\{5;j_0\}$.
\end{proof}

\begin{proposition}

Un morphisme uniforme sans puissance $k$ avec $k \geq 3$ est primitif.

\end{proposition}

\begin{proof}
La preuve est la même que celle de la proposition~\ref{PuisskPrim} mais en utilisant la propriété~\ref{puissancekunif}.
\end{proof}

\section{Morphismes sans carré}

\begin{lemma}\label{Lemme51}
{\rm {(Lemme 4.3 dans~\cite{Hsi2003})}}

Soit $f$ un morphisme de $A^*$ vers $B^*$ et soient $a$ et $b$ deux lettres de $A$.
On suppose qu'il existe deux mots $X$ et $Y$ non simultanément vides de $B^*$ tels que $f(a)=Xf(b)Y$.

Si $X$ est un suffixe non vide de $f(a)$ ou si $Y$ est un préfixe non vide de $f(a)$ alors $f$ n'est pas 
3-sans carré.

\end{lemma}

\begin{proof}

Si $b=a$ alors $f(a)$ est facteur interne de $(f(a))^2$. D'après le lemme~\ref{factint} et la 
remarque~\ref{rem110}, $f(a)$ n'est pas primitif \cad\ $f(a)$ contient un carré.


Si $b \neq a$ et si $X$ est un suffixe non vide de $f(a)$, alors $f(aba)$ contient $Xf(b)Xf(b)$.

Si $b \neq a$ et si $Y$ est un préfixe non vide de $f(a)$ alors $f(aba)$ contient $f(b)Yf(b)Y$.
\end{proof}

\begin{lemma}\label{LemSynch}  {\rm{(Proposition 5.3 dans~\cite{Hsi2003})}}

Soit $f$ un morphisme injectif de $A^*$ vers $B^*$ et soit 
$w$ un mot de $A^+$ de longueur $n \geq2$.
On suppose que $f(w)=u^m$ avec $u \neq \varepsilon$ primitif et $m\geq 2$.

S'il existe deux entiers $1 \leq i_1 < i_2 \leq n$ et un mot $v \in B^*$ tels que
$f(\sub{w}{1..i_1})=u^{\ell_1}v$ et $f(\sub{w}{1..i_2})=u^{\ell_2}v$ pour des entiers $0 \leq \ell_1 < \ell_2 \leq m$
alors $w$ n'est pas primitif.

\end{lemma}

\begin{proof}

Soit $w_r=\sub{w}{1..i_1} \sub{w}{i_2+1..n} \neq \varepsilon$.
On a $f(w_r)=u^{m-(\ell_2-\ell_1)}$
avec $|w_r|<|w|$. De plus,
$f((w_r)^m)=f(w^{m-(\ell_2-\ell_1)})$.
Puisque $f$ est injective, on en déduit que $(w_r)^m=w^{m-(\ell_2-\ell_1)}$.

D'après le corollaire~\ref{CorFine2}, on obtient que $w_r$ et $w$ sont des puissances d'un même mot.
Puisque $|w_r|<|w|$, on en déduit que $w$ n'est pas primitif. 
\end{proof}

\begin{proposition} {\rm \cite{Hsi2003,Mit1996,Mit1997}}

Un morphisme sans carré est primitif.

\end{proposition}

\begin{proof}

Soit $f$ un morphisme sans carré de $A^*$ vers $B^*$.
D'après le lemme~\ref{BifInj}, on a donc $f$ injectif.

Soit $w$ un mot de longueur $n$ tel que $f(w)=u^m$ pour un mot primitif $u$ et un entier $m \geq 2$.

Par contradiction, on suppose que $w$ est primitif. Et on suppose de plus que la longueur de $w$ est minimale. De part le lemme~\ref{LemSynch},
cela signifie par exemple qu'il n'existe pas d'entiers $i \leq j$ tels que 
$|f(\sub{w}{i..j})|$ soit un multiple de $|u|$.

Remarquons aussi que, puisque $f$ est sans carré, 
pour tout lettre $\sub{w}{i}$ de $w$, on a $|f(\sub{w}{i})|<2|u|$ car sinon, $f(\sub{w}{i})$ contiendrait le carré
d'un conjugué de $u$. 

%
%
%

\textit{Cas 1 :} $|u^{j}| < |f(\sub{w}{1})|$  pour un entier $j \geq 1$.

Dans ce cas, comme signalé juste avant, $u^{j}$ étant un préfixe de $f(\sub{w}{1})$ et, 
puisque  ce dernier est sans carré, on a $j=1$.

Il existe un mot $V_1$ de $B^+$ tel que $f(\sub{w}{1})=uV_1$ et $V_1f(\sub{w}{2..n})=u^{m-1}$.

Puisque $V_1$ est un préfixe de $u^{m-1}$ et puisque $|f(\sub{w}{1})|<2|u|$, on a $|V_1| < |u|$ \cad\ 
$V_1$ préfixe de $u$.
Cela implique que $V_1$ est un bord de $f(\sub{w}{1})$.

Si $|f(\sub{w}{n})|<|u|$, il existe un mot $V'$ tel que $f(\sub{w}{1})=uV_1=V'f(\sub{w}{n})V_1$
avec $V_1$ préfixe de $f(\sub{w}{1})$.
D'après le lemme~\ref{Lemme51}, $f$ ne serait pas sans carré~: une
contradiction avec les hypothèses.
On a donc $|f(\sub{w}{n})|>|u|$. De plus, puisque $f(\sub{w}{n}\sub{w}{1})$ contient $uu$ et puisque
$f$ est sans carré, on en déduit que $\sub{w}{n}=\sub{w}{1}$.

Soit $V_2$ le mot tel que $u=V_1V_2$. On obtient que $f(\sub{w}{1})=uV_1=V_1V_2V_1$ 
(on a donc en particulier que $V_2$ est non vide).
Mais $f(\sub{w}{1})=f(\sub{w}{n})$ finit par $u=V_1V_2$. D'après la remarque~\ref{RemLot},
cela signifie que $u$ n'est pas primitif~: ce qui est contraire aux hypothèses.


\vsu
\textit{Cas 2 : } $|u^{j}| < |f(\sub{w}{n})|$  pour un entier $j \geq 1$.

Ce cas se traite exactement de la même façon que le cas 1.

\vsu
\textit{Cas 3 : } $|u| > |f(\sub{w}{1})|$ et $|u| > |f(\sub{w}{n})|$.

Soit $i$ le plus petit entier tel que $|f(\sub{w}{1..i})|>|u|$.
Il existe un suffixe non vide $V_1$ de $u$ et un préfixe non vide $V_2$ de $u^{m-1}$ tels que 
$u=f(\sub{w}{1..i})V_1$, $f(\sub{w}{i})=V_1V_2$ et $u^{m-1}=V_2f(\sub{w}{i+1..n})$.
%
%
%
%

\textit{Cas 3.1 : } $|V_2|>|f(\sub{w}{1})|$.

Soit $\ell$ le plus petit entier tel que $|f(\sub{w}{1..\ell})|>|V_2|$.
Il existe un préfixe $V_2'$ de $f(\sub{w}{\ell})$ qui est suffixe de $V_2$.
Comme $f(\sub{w}{i})f(\sub{w}{\ell})$ contient $(V_2')^2$, et comme $f$ est sans carré, on a nécessairement 
$\sub{w}{\ell}=\sub{w}{i}$.
On a donc $f(\sub{w}{i})=V_1V_2=V_1f(\sub{w}{1..\ell-2})f(\sub{w}{\ell-1})V_2'$ avec 
$V_2'$ préfixe de $f(\sub{w}{i})$.
D'après le lemme~\ref{Lemme51}, $f$ ne serait pas sans carré~: une
contradiction avec les hypothèses.

\textit{Cas 3.2 : } $|V_2| \leq |f(\sub{w}{1})|$.

Puisque $f(\sub{w}{i})f(\sub{w}{1})$ contient $(V_2)^2$ cela implique
que $\sub{w}{1}=\sub{w}{i}$ et donc que $V_2$ est un suffixe de $f(\sub{w}{i})$.

Si $|V_1| > |f(\sub{w}{n})|$, alors il existe un mot non vide $V_1'$ tel que $V_1=V_1'f(\sub{w}{n})$.
On obtient $f(\sub{w}{i})=V_1'f(\sub{w}{n})V_2$. D'après le lemme~\ref{Lemme51}, $f$ ne serait pas sans 
carré~: une contradiction avec les hypothèses.

Si $|V_1| \leq  |f(\sub{w}{n})|$, puisque $f(\sub{w}{n})f(\sub{w}{i})$ contient $(V_1)^2$, et comme $f$ est sans carré, 
on a nécessairement $\sub{w}{n}=\sub{w}{i}$.
Ce qui implique que $V_2V_1=f(\sub{w}{i})=V_1V_2$ et que $f(\sub{w}{i})$ contient un carré~: une contradiction 
avec les hypothèses.
\end{proof}

\begin{cor}
Un morphisme uniforme sans puissance $k (\geq 2)$ est primitif.
\end{cor}

\bibliography{bf.bib}

\end{document}